\documentclass[letterpaper,10pt,conference]{ieeeconf}

\IEEEoverridecommandlockouts
\usepackage{algpseudocode}
\usepackage{algorithm}
\usepackage[compress]{cite}
\usepackage{amsmath}
\usepackage{physics}
\usepackage{amssymb}
\usepackage{tikz}
\usepackage{psfrag}
\usepackage{pstool}

\usetikzlibrary{positioning}
\usepackage{tikzpeople}

\newtheorem{theorem}{\bf Theorem}
\newtheorem{proposition}{\bf Proposition}

\newtheorem{corollary}{\bf Corollary}

\newtheorem{remark}{\bf Remark}
\newtheorem{lemma}{\bf Lemma}
\newtheorem{definition}{\bf Definition}

\makeatletter
\renewcommand*{\@opargbegintheorem}[3]{\trivlist
  \item[\hskip \labelsep{\bfseries #1\ #2}] \textbf{(#3)}\ \itshape}
\makeatother

\let\chap\S

\newcommand{\Y}{\mathbb{Y}}
\newcommand{\X}{\mathbb{X}}
\newcommand{\Z}{\mathbb{Z}}
\newcommand{\G}{\mathcal{G}}

\newcommand{\U}{\mathcal{U}}
\newcommand{\B}{\mathcal{B}}
\newcommand{\F}{\mathcal{F}}
\renewcommand{\S}{\mathcal{S}}
\renewcommand{\H}{\mathcal{H}}
\renewcommand{\L}{\mathcal{L}}
\renewcommand{\P}{\mathcal{P}}

\allowdisplaybreaks

\begin{document}

\title{Measuring Quantum Information Leakage Under Detection Threat} 

\author{ 
Farhad Farokhi and Sejeong Kim
\thanks{
The authors are with the Department of Electrical and Electronic Engineering at the University of Melbourne. 
}
\thanks{e-mails:\{farhad.farokhi, sejeong.kim\}@unimelb.edu.au}
\vspace{-5mm}
}

\maketitle

\thispagestyle{plain}
\pagestyle{plain}

\begin{abstract} Gentle quantum leakage is proposed as a measure of information leakage to arbitrary eavesdroppers that aim to avoid detection. Gentle (also sometimes referred to as weak or non-demolition) measurements are used to encode the desire of the eavesdropper to evade detection. The gentle quantum leakage meets important axioms proposed for measures of information leakage including positivity, independence, and unitary invariance. Global depolarizing noise, an important family of physical noise in quantum devices, is shown to reduce gentle quantum leakage (and hence can be used as a mechanism to ensure privacy or security). A lower bound for the gentle quantum leakage based on asymmetric approximate cloning is presented. This lower bound relates information leakage to mutual incompatibility of quantum states. A numerical example, based on the encoding in the celebrated BB84 quantum key distribution algorithm, is used to demonstrate the results.
\end{abstract}

\section{Introduction}
Quantum cryptography and key distribution, such as the celebrated BB84 algorithm~\cite{bennett2014quantum}, often rely on two fundamental features of quantum mechanical systems: no cloning theorem and post-measurement state collapse~\cite{nielsen2010quantum}. No cloning theorem states that it is impossible to copy a generic quantum state and therefore an eavesdropper must take direct measurements of the underlying quantum mechanical system that is used for communication (e.g., polarization of photons). Post-measurement state collapse implies that upon observing the quantum mechanical system by the eavesdropper, its state will stochastically and irreparably change, which can be used by the sender or receiver to identify presence of an intruder~\cite{bennett2014quantum,ekert1994eavesdropping}. Upon detection of the eavesdropper, the authorized users can abandon communication or switch medium to avoid the eavesdropper. This motivates investigating the interplay between information extraction by the eavesdropper (by taking informative measurements) versus its detection by the authorized users (by post-measurement state collapse) to determine the optimal trade-off. 

The problem of investigating the trade-off between measurement informativeness and state collapse is not new. It has attracted attention of physicist in the past~\cite{fuchs1996quantum,fuchs1998information, fuchs2001information,banaszek2006information}. However, they may not be fit for purpose due to two main reasons.  First, they assume that the eavesdropper is interested in estimating the entirety of the classical information that is encoded in the quantum system. This restrict our modeling of the eavesdropper. A better approach to investigating eavesdroppers in security and privacy literature is to consider the worst-case information leakage over anything that they try to guess or estimate~\cite{issa2019operational,Farokhi_PRA,liao2019tunable, farokhi2021measuring}. This way, we do not underestimate the eavesdropper and consider a more general threat model. Second, mutual information is not a good measure for information leakage~\cite{issa2019operational}. It was stated in~\cite{fuchs1996quantum} that ``mutual information ... may not be the quantities that are most relevant to applications in quantum cryptography.'' The motivation for mutual information is rooted in data compression and transmission with vanishing error, while an eavesdropper may be content with guessing most likely outcomes, e.g., as an starting point for phishing attacks. Thus even a modest increase in the probability of successfully guessing some attributes of the encoded data can lead to devastating consequences.

In this paper, we develop a notion for information leakage against an arbitrary eavesdropper (one whose intention is not entirely clear to us) that aims to avoid detection. We use gentle measurements, utilized in~\cite{aaronson2019gentle}, to encode the desire of the eavesdropper to evade detection. A measurement is called gentle if the post-measurement state  remains in proximity of the state prior to measurement with high probability (over outcomes of measurement). The measure of proximity, i.e., the magnitude of the change caused by the measurement, in the gentle measurement framework can be tied to the probability of detection by the Bayesian quantum hypothesis testing~\cite{fuchs1998information}. This notion of quantum information leakage, referred to as \textit{gentle quantum leakage}, measures the worst-case multiplicative increase in the probability of correctly guessing any deterministic or randomized function of the private classical data with and without access to the quantum system encoding the-said data. This way, we search over all possible goals for the eavesdropper and do not restrict our analysis. 

We derive a semi-explicit formula for the gentle quantum measurement based on the Sibson information of order infinity (an extension of the mutual information~\cite{7308959}). Furthermore, we prove that the gentle quantum leakage meets important axioms for a measure of information leakage including positivity (i.e., gentle quantum leakage is always greater than or equal to zero), independence (i.e., gentle quantum leakage is zero if the quantum state is independent of the classical data), and unitary invariance (gentle quantum leakage remains unchanged by transforming the quantum encoding of the classical data by a unitary channel). These axioms have been formulated for classical and quantum notions of information leakage~\cite{Farokhi_PRA,issa2019operational, farokhi2021measuring,muller2013quantum,10106314906367}. We provide upper bounds for the gentle quantum leakage based on the number of the possibilities of the classical data (i.e., the cardinality of its support set) and the dimension of the quantum system used for encoding the data. We prove that global depolarizing noise, an important family of physical noise in quantum devices~\cite{PhysRevE104035309}, reduces the gentle quantum leakage and therefore, can be used as a mechanism for mitigating security and privacy threats, albeit at the risk of sacrificing  utility~\cite{farokhi2024barycentric,nuradha2023quantum}. We also provide a lower bound for gentle quantum leakage based on asymmetric approximate cloning~\cite{mitra2021optimal}. 
A numerical example, based on the encoding in the BB84~\cite{bennett2014quantum}, is used to demonstrate the results.

The remainder of the paper is organized as follows. Section~\ref{sec:preliminary} presents some preliminary material and notations on quantum systems and information theory. Gentle quantum leakage is presented in Section~\ref{sec:gentle_leakage}. A lower bound based on approximate cloning is presented in Section~\ref{sec:cloning}. Numerical results are presented in Section~\ref{sec:numerical}. Finally, Section~\ref{sec:conc} concludes the paper and presents some directions for future work.

\section{Preliminaries} \label{sec:preliminary}
Finite-dimensional Hilbert spaces are denoted by $\H$. The set of linear operators from $\H_A$ to $\H_B$ is denoted by $\L(\H_A,\H_B)$. When $\H_A=\H_B=\H$, with slight abuse of notation, we write $\L(\H)$ instead of $\L(\H_A,\H_B)$. The set of positive semi-definite linear operators on Hilbert space $\H$ is denoted by $\P(\H)\subset\L(\H)$. We write $A\geq 0$ when $A\in\P(\H)$ and $A\geq B$ if $A-B\geq 0$. Furthermore, the set of positive semi-definite linear operators on Hilbert space $\H$ with unit trace, also known as density operators, is denoted by $\S(\H)\subset\P(\H)\subset\L(\H)$. We use lower case Greek letters, such as $\rho$ and $\sigma$, to denote density operators. Conventionally, density operators are used to model states of quantum systems~\cite[\chap\,4]{wilde2013quantum}. We define the (normalized) trace distance between any two quantum states $\rho,\sigma\in\S(\H)$ as
\begin{align*}
    \|\rho-\sigma \|_{\trace}:=\frac{1}{2}\trace(|\rho-\sigma|),
\end{align*}
where $|M|=\sqrt{M^\dag M}$ for any operator $M\in\L(\H)$. An important property of the trace distance is that it is unitary invariant, i.e., $\|U(\rho-\sigma )V^\dag\|_{\trace}=\|\rho-\sigma \|_{\trace}$ for all unitary operators $U$ and $V$~\cite[Property~9.1.4]{wilde2013quantum}. Note that operator $U\in\L(\H)$ is unitary if $U^\dag U=U U^\dag=I$. For any two linear operators $A,B\in\L(\H)$, $[A,B]=AB-BA$ denotes their commuter. Linear operators $A,B\in\L(\H)$ commute if and only if $[A,B]=0$. 

Generalized quantum measurements are modelled using the Positive Operator Valued Measure (POVM) framework. POVMs can model the von Neumann, i.e., projection-based, measurements and their extensions, e.g., when the measurement involves interaction with ancillary systems~\cite{wiseman2010quantum}. A POVM, with the set of possible outcomes $\Y$, is a set of positive semi-definite operators $\F=\{F_y\}_{y\in\Y}\subseteq \P(\H)$ such that $\sum_{y\in\Y}F_y=I$. The probability of observing measurement outcome $y\in\Y$ on a quantum system with state $\rho\in\S(\H)$ is $\trace(\rho F_y)$. This is commonly called the Born's rule. Note that POVMs do not explicitly consider the post-measurement state of the quantum system, i.e., post-measurement collapse of the state. To also model the post-measurement state collapse, we need to consider an implementation $\B=\{B_y\}_{y\in\Y}$ of POVM $\F=\{F_y\}_{y\in\Y}$ such that $F_y=B_y^\dag B_y$ for all $y\in\Y$. For any operator $B$, $B^\dag$ denotes its complex conjugate transpose. Using implementation $\B=\{B_y\}_{y\in\Y}$, the post-measurement state is given by 
\begin{align*}
    \rho_{\B|y}=\frac{B_y \rho B_y^\dag}{\trace(B_y \rho B_y^\dag)}=\frac{B_y \rho B_y^\dag}{\trace(\rho F_y)},
\end{align*}
if the outcome $y\in\Y$ is observed and the pre-measurement state is $\rho\in\S(\H)$. It is postulated that it is impossible to observe a quantum system, i.e., take measurements, without disturbing its state via the so-called post-measurement state collapse~\cite{fuchs1996quantum}. This is a feature that often underlies quantum cryptography and key distribution~\cite{bennett2014quantum}. However, some measurements disturb the state more than others. A family of measurements that do not significantly disturb the state are called gentle~\cite{aaronson2019gentle} or weak measurements~\cite{fuchs1996quantum}. We use these measurements to develop a notion of information leakage when the eavesdropper aims to be undetected. 

\begin{definition}[Gentle Measurements]
    Consider the set of states $S\subseteq\S(\H)$ and constants $\alpha,\delta\in[0,1]$. 
    
    The POVM $\F=\{F_y\}_{y\in\Y}$ is $(\alpha,\delta)$-weakly gentle on $S$ if it possesses at least one implementation $\B=\{B_y\}_{y\in\Y}$, i.e., $F_y=B_y^\dag B_y$ for all $y\in\Y$, such that 
    \begin{align} \label{eqn:def_probability_gentle}
        \mathbb{P}\left\{\|\rho_{\B|Y}-\rho\|_{\trace}\leq \alpha,\forall \rho\in S\right\}\geq 1-\delta,
    \end{align}
    where the probability is taken with respect to measurement outcomes, i.e., random variable $Y\in\Y$ denotes the measurement outcome. The set of all $(\alpha,\delta)$-weakly gentle POVMs on $S$ is denoted by $\G_{(\alpha,\delta)}(S)$. 
    
    The POVM $\F=\{F_y\}_{y\in\Y}$ is $(\alpha,\delta)$-strongly gentle on $S$ if~\eqref{eqn:def_probability_gentle} holds for {\em all} implementations $\B=\{B_y\}_{y\in\Y}$. The set of all $(\alpha,\delta)$-strongly gentle POVMs on $S$ is denoted by $\overline{\G}_{(\alpha,\delta)}(S)$. 
\end{definition}

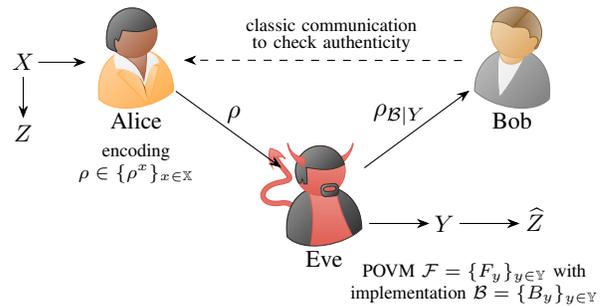
\begin{figure}
    \centering
    \begin{tikzpicture}[]
    \node[bob, minimum width=1cm, mirrored] (bob) {\small Bob};
    \node[alice, minimum width=1cm, left=4cm of bob] (alice) {\small Alice};
    \node[devil, minimum width=1cm, below left=5mm and 1.5cm of bob] (eve) {\small Eve};
    \node[] at (-.5,-3.){
    \begin{minipage}{3.4cm}
    \scriptsize\centering
    POVM $\F=\{F_y\}_{y\in\Y}$ with implementation $\B=\{B_y\}_{y\in\Y}$
    \end{minipage}
    };
    \draw[-{Stealth}] (-6.3,-.055) -- (alice);
    \node[] at (-6.5,-.05) {\small $X$};
    \draw[-{Stealth}] (-6.5,-.25) -- (-6.5,-.8) ;
    \node[] at (-6.5,-1) {\small $Z$};
    \draw[-{Stealth}] (alice) -- (eve);
    \draw[-{Stealth}] (eve) -- (bob);
    \draw[-{Stealth},dashed] (bob) -- (alice);
    \draw[-{Stealth}] (-1.9,-2.2) -- (-1.1,-2.2);
    \node[] at (-.9,-2.2) {\small $Y$};
    \draw[-{Stealth}] (-.7,-2.2) -- (.1,-2.2);
    \node[] at (.3,-2.15) {\small $\widehat{Z}$};
    \node[] at (-3.7,-0.75) {$\rho$};
    \node[] at (-1.5,-0.75) {$\rho_{\B|Y}$};
    \node[] at (-2.4,0.3) {\scriptsize 
    \begin{minipage}{2.4cm}\centering 
    classic communication \\[-.2em] to check authenticity
    \end{minipage}
    };
    \node[] at (-5,-1.4) {\scriptsize 
    \begin{minipage}{2.4cm}\centering 
    encoding $\rho\in\{\rho^x\}_{x\in\X}$
    \end{minipage}
    };
    \end{tikzpicture}
    \vspace{-3mm}
    \caption{Communication schematic between Alice, Bob, and Eve.}
    \vspace{-5mm}
    \label{fig:ABE}
\end{figure}

\section{Maximal Leakage with Gentle Measurements} \label{sec:gentle_leakage}
We use discrete random variable $X\in\X$ to model the classical data that must be protected. Without loss of generality, we assume that $p_X(x)=\mathbb{P}\{X=x\}>0$ for all $x\in\X$; otherwise, we can trim the set $\X$ to only contain elements with non-zero probability. For each $X=x\in\X$, Alice prepares quantum system $A$ in mixed state $\rho^x\in\P(\H)$, i.e., encodes $x$ as $\rho^x$. The ensemble $\mathcal{E}:=\{p_X(x),\rho^x\}_{x\in\X}$ models Alice's quantum encoding of the classical data. Alice intends to communicate this quantum system to Bob as in Figure~\ref{fig:ABE}. However, this communication maybe intercepted by an eavesdropper, Eve. From the perspective of someone who does not know the realization of classical data $X$, i.e., Bob and Eve, the state is given by the expected density operator $\rho=\mathbb{E}\{\rho^X\}=\sum_{x\in\X}p_X(x)\rho^x$.

Eve wants to estimate a possibly randomized discrete function of the classical data $X$, denoted by discrete random variable $Z$. The nature and composition of this random variable, i.e., the target of the eavesdropping attack, is not known by Alice or Bob. This can only be done by taking measurements of the quantum system, which is modelled by POVM $\F=\{F_y\}_{y\in\Y}$. Let discrete random variable $Y\in\Y$ denote the outcome of the measurement. By Born's rule, $\mathbb{P}\{Y=y\,|\,X=x\}=\trace(\rho^x F_y)$ for all $x\in\X$ and $y\in\Y$. Eve uses the measurement outcome to take a one-shot\footnote{Number of guesses is immaterial in measuring information leakage~\cite{Farokhi_PRA}.} guess of the random variable $Z$ denoted by the random variable $\widehat{Z}$. Maximal quantum leakage, defined in~\cite{Farokhi_PRA}, measures the multiplicative increase of the probability of correctly guessing the realization of any random variable $Z$ with and without access to the measurement outcome $Y$. 

\begin{definition}[Maximal Quantum Leakage~\cite{Farokhi_PRA}] \label{def:qml} Maximal quantum leakage from random variable $X$ through quantum system $A$ is 
\begin{subequations}\label{eqn:def_qml}
    \begin{align}
    \mathcal{Q}(X\!\rightarrow \!A)_{\rho}
    :=&\sup_{\{F_y\}_{y\in\Y}}\sup_{Z,\widehat{Z}} \log \!\left(\!\frac{\mathbb{P}\{Z=\widehat{Z}\}}{\displaystyle \max_{z\in\mathbb{Z}}\mathbb{P}\{Z=z\}} \!\right)\label{eqn:qml:original_def}\\
    =&\sup_{\{F_y\}_{y\in\Y}} \log\!\left(\sum_{y\in\mathbb{Y}} \max_{x\in\mathbb{X}} \trace(\rho^x F_y) \!\right)\!,
\end{align}
\end{subequations}
    where, in~\eqref{eqn:qml:original_def}, the inner supremum is taken over all random variables $Z,\widehat{Z}\in\Z$ with arbitrary finite support set $\Z$ and the outer supremum is taken over all POVMs $\F=\{F_y\}_{y\in\Y}$ with arbitrary finite outcome set $\Y$. 
\end{definition}

For maximal quantum leakage, defined above, we search over all measurements to find one that results in the most information gain, measured by multiplicative increase in the probability of correctly guessing a secret. This would imply that the  post-measurement state could be arbitrarily disturbed. This disturbance can be used by Alice or Bob to identify presence of an eavesdropper (in which case they can cease communication or switch to a different medium). For instance, the famous quantum key distribution algorithm of BB84~\cite{bennett2014quantum} uses the post-measurement state disturbance, no-cloning theorem in quantum systems, and access to a classical communication channel to identify presence of an eavesdropper. Therefore, it is postulated that Eve must use gentle or weak measurements to avoid getting caught~\cite{fuchs1996quantum}. Note that the use of the trace-distance in the definition of gentle measurements can be justified if the Alice or Bob use the Bayesian hypothesis-testing to identify presence of an eavesdropper~\cite{fuchs1998information}. Other definitions, e.g., using fidelity~\cite{fuchs1996quantum,fuchs1998information}, can be also used. However, they can be translated into the definition relying on trace distance due to the relationship between fidelity and trace distance~\cite[Theorem~9.3.1]{wilde2013quantum}.

\begin{definition}[Gentle Quantum Leakage] \label{def:gqml} For any $\alpha,\delta\in[0,1]$,  $(\alpha,\delta)$-{\em weakly} gentle quantum leakage from random variable $X$ through quantum system $A$ is 
\begin{align}
    \mathcal{L}_{(\alpha,\delta)}&(X\rightarrow A)_{\rho}\nonumber\\&:=\sup_{\{F_y\}_y\in\G_{(\alpha,\delta)}(S)}\sup_{Z,\widehat{Z}} \log_2 \left(\frac{\mathbb{P}\{Z=\widehat{Z}\}}{\displaystyle \max_{z\in\mathbb{Z}}\mathbb{P}\{Z=z\}} \right)\!,
\end{align}
    where the outer supremum is now taken over the set of all $(\alpha,\delta)$-{\em weakly} gentle measurements on $S:=\{\rho^x\}_{x\in\X}$ denoted by $\G_{(\alpha,\delta)}(S)$. Similarly, $(\alpha,\delta)$-{\em strongly} gentle quantum leakage is 
\begin{align}
    \overline{\mathcal{L}}_{ (\alpha,\delta)}& (X\rightarrow A)_{\rho}\nonumber\\&:=\sup_{\{F_y\}_y\in\overline{\G}_{(\alpha,\delta)}(S)}\sup_{Z,\widehat{Z}} \log_2 \left(\frac{\mathbb{P}\{Z=\widehat{Z}\}}{\displaystyle \max_{z\in\mathbb{Z}}\mathbb{P}\{Z=z\}} \right)\!,
\end{align}   
where the outer supremum is now taken over the set of all $(\alpha,\delta)$-{\em strongly} gentle measurements.
\end{definition}

This notion of information leakage keeps the disturbance on the state caused by measurement below $\alpha$ with probability of at least $1-\delta$ (with respect to measurement outcomes). This way, we can investigate the trade-off between information leakage and Eve's chance of getting caught by Alice or Bob. 

\begin{corollary} \label{cor:bound}
    The following properties hold:
    \begin{itemize}
        \item $\mathcal{L}_{(\alpha,\delta)}(X\rightarrow A)_{\rho}$ and $\overline{\mathcal{L}}_{(\alpha,\delta)}(X\rightarrow A)_{\rho}$ are non-decreasing in $\alpha,\delta\in[0,1]$;
        \item $\mathcal{L}_{(\alpha,\delta)}(X\rightarrow A)_{\rho}\geq \overline{\mathcal{L}}_{(\alpha,\delta)}(X\rightarrow A)_{\rho}$ for all $\alpha,\delta\in[0,1]$;
        \item $\mathcal{L}_{1,\delta}(X\rightarrow A)_{\rho}=\mathcal{L}_{\alpha,1}(X\rightarrow A)_{\rho}=\overline{\mathcal{L}}_{1,\delta}(X\rightarrow A)_{\rho}=\overline{\mathcal{L}}_{\alpha,1}(X\rightarrow A)_{\rho}=\mathcal{Q}(X\!\rightarrow \!A)_{\rho}$ for all $\alpha,\delta\in[0,1]$.
    \end{itemize}
\end{corollary}

\begin{proof}
    These properties can be immediately seen from that $\overline{\G}_{(\alpha,\delta)}(\{\rho^x\}_{x\in\X})\subseteq\G_{(\alpha,\delta)}(\{\rho^x\}_{x\in\X})$, and that $\G_{(\alpha,\delta)}(\{\rho^x\}_{x\in\X})\subseteq \G_{(\alpha',\delta')}(\{\rho^x\}_{x\in\X})$ and $\overline{\G}_{(\alpha,\delta)}(\{\rho^x\}_{x\in\X})\subseteq \overline{\G}_{(\alpha',\delta')}(\{\rho^x\}_{x\in\X})$ if $\alpha'\geq \alpha$ and $\delta'\geq \delta$. Also, as $\alpha\rightarrow 1$ or $\delta\rightarrow 1$, $\G_{(\alpha,\delta)}(\{\rho^x\}_{x\in\X})$ and $\overline{\G}_{(\alpha,\delta)}(\{\rho^x\}_{x\in\X})$ converges to the set of all POVMs.
\end{proof}

Using the same line of reasoning as in the maximal quantum leakage in~\cite{Farokhi_PRA}, we can derive a semi-explicit formula for the gentle quantum leakage. This derivation follows from eliminating the first supremum on random variables $Z$ and $\widehat{Z}$ using their classical counterpart, i.e., maximal leakage, in~\cite{issa2019operational}. The following result reformulates the gentle quantum leakage using the Sibson mutual information of order infinity, which is a generalized notion of information in classical information theory and can be reviewed in~\cite{issa2019operational,7308959}.

\begin{corollary} \label{cor:mql} The gentle quantum leakage is 
    \begin{align} \label{eqn:maxim_leakage_formula}
    \mathcal{L}_{(\alpha,\delta)}(X\rightarrow A)_{\rho}&=\sup_{\{F_y\}_y\in\G_{(\alpha,\delta)}(\{\rho^x\}_{x\in\X})}I_\infty(X;Y),\\
    \overline{\mathcal{L}}_{(\alpha,\delta)}(X\rightarrow A)_{\rho}&=\sup_{\{F_y\}_y\in\overline{\G}_{(\alpha,\delta)}(\{\rho^x\}_{x\in\X})}I_\infty(X;Y),
 \end{align}
 where $I_\infty(X;Y)$ is the Sibson mutual information  of order infinity between random variables $X$ and $Y$ computed as
 \begin{align*}
      I_\infty(X;Y):=& \log_2\left(\sum_{y\in\mathbb{Y}} \max_{x\in\mathbb{X}:p_X(x)>0} \mathbb{P}\{Y=y\,|\,X=x\} \right)\\
      =&\log_2\left(\sum_{y\in\mathbb{Y}} \max_{x\in\mathbb{X}} \trace(\rho^x F_y) \right).
 \end{align*}
\end{corollary}

\begin{proof}
    The proof follows from applying Theorem~1 in~\cite{issa2019operational} to resolve the supremum on $Z$ and $\widehat{Z}$.
\end{proof}

In what follows, we prove that the gentle quantum leakage is an appropriate notion of information leakage by demonstrating that it meets important axioms~\cite{Farokhi_PRA,issa2019operational,farokhi2021measuring,muller2013quantum,10106314906367}, such as positivity (i.e., gentle quantum leakage is always greater than or equal to zero), independence (i.e., gentle quantum leakage is zero if the quantum state is independent of the classical data), and unitary invariance (gentle quantum leakage remains unchanged by transforming the quantum encoding of the classical data by a unitary channel). 

\begin{proposition}[Positivity and Independence] \label{prop:independence}
    The following properties hold:
    \begin{itemize}
        \item $\mathcal{L}_{(\alpha,\delta)}(X\rightarrow A)_{\rho}\geq 0$ with equality if and only if $\rho^x=\rho^{x'}$ for all $x,x'\in\X$;
        \item $\overline{\mathcal{L}}_{(\alpha,\delta)}(X\rightarrow A)_{\rho} \geq 0$ with equality if $\rho^x=\rho^{x'}$ for all $x,x'\in\X$.
    \end{itemize}
\end{proposition}

\begin{proof}
     Note that $I_\infty(X;Y)\geq 0$~\cite{issa2019operational,7308959} and therefore, 
     \begin{align*}
         \mathcal{L}_{(\alpha,\delta)}(X\rightarrow A)_{\rho}&=\sup_{\{F_y\}_y\in \G_{(\alpha,\delta)}(\{\rho^x\}_{x\in\X})} I_\infty(X;Y)\geq 0,\\
         \overline{\mathcal{L}}_{ (\alpha,\delta)}(X\rightarrow A)_{\rho}&=\sup_{\{F_y\}_y\in \overline{\G}_{(\alpha,\delta)}(\{\rho^x\}_{x\in\X})} I_\infty(X;Y)\geq 0.
     \end{align*}

    Corollary~\ref{cor:bound} implies that  $0\leq \overline{\mathcal{L}}_{(\alpha,\delta)}(X\rightarrow A)_{\rho}\leq \mathcal{L}_{(\alpha,\delta)}(X\rightarrow A)_{\rho}\leq \mathcal{Q}(X\rightarrow A)_{\rho}=0$ if $\rho^x=\rho^{x'}$ for all $x,x'\in\X$~\cite{Farokhi_PRA}. 

     For the reverse, assume that $\mathcal{L}_{(\alpha,\delta)}(X\rightarrow A)_{\rho}=0$. Following~\cite[Lemma~1]{issa2019operational}, it must be that $X$ and $Y$ are independent for all POVMs $\{F_y\}_y\in \G_{(\alpha,\delta)}(\{\rho^x\}_{x\in\X})$. Thus $
        \trace(\rho^x F_y)=\mathbb{P}\{Y=y\,|\,X=x\}=\mathbb{P}\{Y=y\,|\,X=x'\}=\trace(\rho^{x'} F_y)$,
    where the second equality follows from statistical independence of $X$ and $Y$. Therefore, $\trace((\rho^x-\rho^{x'}) F_y)=0$ for all $(\alpha,\delta)$-weakly gentle  measurements and all $x,x'\in\X$. We use this to show that $\rho^x-\rho^{x'}=0$ for all all $x,x'\in\X$. Assume that there exists $x,x'\in\X$ such that $\rho^x-\rho^{x'}\neq 0$. Since $\rho^x-\rho^{x'}$ is Hermitian, we can diagonalize it as $\rho^x-\rho^{x'}=\sum_{i}\lambda_i \ket{i}\bra{i}$. Define $L_+:=\sum_{i:\lambda_i\geq 0}\lambda_i \ket{i}\bra{i}$ and $L_-:=\sum_{i:\lambda_i< 0}\lambda_i \ket{i}\bra{i}$. By construct $L_+$ and $L_-$ are both semi-positive definite and $L_+L_-=L_-L_+=0$. Without loss of generality, assume that $L_+>0$ (otherwise, we can swap $x$ and $x'$). Consider 
    \begin{align*}
        B_+&=\sqrt{\frac{1-2\epsilon^2}{2}}I+\epsilon L_+,\\
        B_-&=\sqrt{\frac{1-2\epsilon^2}{2}}I-\epsilon L_+,\\
        B_0&=\sqrt{2}\epsilon(I- L_+^2)^{1/2}.
    \end{align*}
    Lemma~\ref{lemma:gentle_construct} in Appendix~\ref{app:useful_lemma} shows that there exists $\epsilon'>0$ such that, for all $0\leq \epsilon\leq \epsilon'$, POVM $\{F_+,F_-,F_0\}$ with $F_y=B_yB_y^\dag$ with $y\in\{+,-,0\}$ is $(\alpha,\delta)$-weakly gentle on $\{\rho^x\}_{x\in\X}$. We have
    \begin{align*}
        \trace((\rho^x-\rho^{x'}) F_+)
        =&\trace(B_+\rho^xB_+)-\trace(B_+\rho^{x'}B_+)\\
        =&\sqrt{\frac{1-2\epsilon^2}{2}}\epsilon\trace((\rho^x-\rho^{x'})L_+)\\
        &+\sqrt{\frac{1-2\epsilon^2}{2}}\epsilon\trace(L_+(\rho^x-\rho^{x'}))\\
        &+\epsilon^2\trace(L_+^2(\rho^x-\rho^{x'}))\\
        =&\sqrt{\frac{1-2\epsilon^2}{2}}\epsilon\trace((L_+-L_-)L_+)\\
        &+\sqrt{\frac{1-2\epsilon^2}{2}}\epsilon\trace(L_+(L_+-L_-))\\
        &+\epsilon^2\trace(L_+^2(L_+-L_-))\\
        =&2\sqrt{\frac{1-2\epsilon^2}{2}}\epsilon\trace(L_+^2)+\epsilon^2\trace(L_+^3)\\
        > & 0.
    \end{align*}
    This is in contradiction with that $\trace((\rho^x-\rho^{x'}) F_y)=0$ for all $(\alpha,\delta)$-weakly gentle measurements. Therefore, it must be that $\rho^x-\rho^{x'}=0$.
\end{proof}

In the next result, we show that the strongly gentle quantum leakage remains unchanged by application of unitary quantum channels on the quantum encoding of the classical data. Unitary operators are often used for quantum computing~\cite{nielsen2010quantum} and quantum machine learning~\cite{Benedetti_2019}. Therefore, this result demonstrates that quantum information leakage, measured by strongly gentle quantum leakage, cannot increase by employing arbitrary quantum computing routines (which also contains classical computing routines as a subset).

\begin{proposition}[Unitary Invariance] \label{prop:unitary_invariance}
    For any unitary channel $\U(\rho)=U\rho U^\dag$ with unitary operator $U\in\L(\H)$, 
    \begin{align*}
        \overline{\mathcal{L}}_{ (\alpha,\delta)}(X\rightarrow A)_{\U(\rho)}= \overline{\mathcal{L}}_{(\alpha,\delta)}(X\rightarrow A)_{\rho}.
    \end{align*}
\end{proposition}

\begin{proof}
    From Corollary~\ref{def:qml}, we have
    \begin{align*}
        \overline{ \mathcal{L}}_{(\alpha,\delta)}&(X\rightarrow A)_{\U(\rho)}\\&\hspace{-.2in}=\sup_{\{F_y\}_y\in\overline{\G}_{(\alpha,\delta)}(\{\U(\rho^x)\}_{x\in\X})} &\hspace{-.15in}\log_2\left(\sum_{y\in\mathbb{Y}} \max_{x\in\mathbb{X}} \trace(\mathcal{U}(\rho^x) F_y) \!\!\right)\!\!.
    \end{align*}
    Consider POVM $\F=\{F_y\}_{y\in\Y}$. It is easy to see that $\trace(\U(\rho^x)F_y)= \trace\left(\rho^x F'_y\right)$ with $F'_y:=U^\dag F_yU$ for all $y\in\Y$. Note that $\{F'_y\}_{y\in\Y}$ is a POVM~\cite{Farokhi_PRA}. Let
    \begin{align*} 
        \G':=\{\{F'_y\}_{y\in\Y}\,|\,&F'_y=U^\dag F_yU,\forall y\in\Y,\nonumber\\
        &\{F_y\}_{y\in\Y}\in\overline{\G}_{(\alpha,\delta)}(\{\U(\rho^x)\}_{x\in\X})\}. 
    \end{align*}
    As a result,
    \begin{align}
        \overline{ \mathcal{L}}_{(\alpha,\delta)}&(X\rightarrow A)_{\U(\rho)}\!=\!\!\!\!\sup_{\{F'_y\}_y\in\G'} \!\!\!\!\log_2\!\!\left(\sum_{y\in\mathbb{Y}} \max_{x\in\mathbb{X}} \trace(\rho^x F'_y) \!\!\right)\!\!.
        \label{eqn:proof:2}
    \end{align}
    
    Fact \#1 ($\G'\subseteq \overline{\G}_{(\alpha,\beta)}(\{\rho^x\}_{x\in\X})$): Pick any $\{F'_y\}_{y\in\Y}\in\G'$ with any implementation $\{B'_y\}_{y\in\Y}$, i.e.,  ${B'_y}^\dag B'_y=F'_y$ for all $y\in\Y$. By construct, there must exists $\{F_y\}_{y\in\Y}\in \overline{\G}_{(\alpha,\delta)}(\{\U(\rho^x)\}_{x\in\X})$ such that $F'_y=U^\dag F_yU$ for all $y\in\Y$.  For all $x\in\X$, we have
    \begin{align*}
         \hspace{.5in}&\hspace{-.5in} \left\|\frac{B'_y \rho^x {B'_y}^\dag}{\trace(B'_y\rho^x {B'_y}^\dag)} -\rho^x\right\|_{\trace}
         \\=& 
         \left\|\frac{(B'_y U^\dag) U \rho^x U^\dag (B'_y U^\dag )^\dag}{\trace((B'_y U^\dag) U \rho^x U^\dag (B'_y U^\dag )^\dag)} -\rho^x\right\|_{\trace}\\
         =& 
         \left\|\frac{(B'_y U^\dag ) \U(\rho^x) (B'_y U^\dag )^\dag}{\trace(\U(\rho^x) UF'_yU^\dag)} -\rho^x\right\|_{\trace}\\
         =& 
         \left\|U\left( \frac{(B'_y U^\dag ) \U(\rho^x) (B'_y U^\dag )^\dag}{\trace(\U(\rho^x)F_y)} -\rho^x \right)U^\dag\right\|_{\trace}\\
         =& 
         \left\|\frac{(UB'_yU^\dag) \U(\rho^x) (UB'_yU^\dag)^\dag}{\trace(\U(\rho^x)F_y)} -\U(\rho^x)\right\|_{\trace}\\
         \leq & \alpha,
    \end{align*}
    with probability of at least $1-\delta$ because $\{UB'_yU^\dag\}_y$ is an implementation of $\{F_y\}_y$, i.e., $(UB'_yU^\dag)^\dag (UB'_yU^\dag)=U{B'_y}^\dag U^\dag UB'_yU^\dag=UF'_y U^\dag=F_y$ for all $y\in\Y$, and $\{F_y\}_y\in\overline{\G}_{(\alpha,\beta)}(\{\U(\rho^x)\}_{x\in\X})$. Therefore, $\{F'_y\}_{y\in\Y}\in\overline{\G}_{(\alpha,\beta)}(\{\rho^x\}_{x\in\X})$. As a result, $\G'\subseteq \overline{\G}_{(\alpha,\beta)}(\{\rho^x\}_{x\in\X})$. 

    Fact \#2 ($\overline{\G}_{(\alpha,\beta)}(\{\rho^x\}_{x\in\X})\subseteq \G'$): Pick any $\{F'_y\}_{y\in\Y}\in \overline{\G}_{(\alpha,\beta)}(\{\rho^x\}_{x\in\X})$. Define $F_y=U F'_y U^\dag$ for all $y\in\Y$. Clearly, $F'_y=U^\dag F_y U$ because $U$ is unitary. Let $\{B_y\}_{y\in\Y}$ be any implementation of $\{F_y\}_{y\in\Y}$.  For all $x\in\X$, we have
    \begin{align*}
        \hspace{.5in}&\hspace{-.5in} \left\|\frac{B_y \U(\rho^x) B_y^\dag}{\trace(B_y \U(\rho^x) B_y^\dag)} -\U(\rho^x) \right\|_{\trace}
        \\ =&
        \left\|\frac{B_y U \rho^x U^\dag B_y^\dag}{\trace(B_y \U(\rho^x) B_y^\dag)} -U\rho^x U^\dag \right\|_{\trace}\\
        =& 
        \left\|U\left(\frac{(U^\dag B_y U) \rho^x (U^\dag B_y U)^\dag}{\trace( \rho^x U^\dag  F_y U)} -\rho^x \right)U^\dag \right\|_{\trace}\\
        =& 
        \left\|\frac{(U^\dag B_y U) \rho^x (U^\dag B_y U)^\dag}{\trace(\rho^x  F'_y)} -\rho^x \right\|_{\trace}\\
        \leq & \alpha,
    \end{align*}
    with probability of at least $1-\delta$ because $\{U^\dag B_y U\}_{y\in\Y}$ is an implementation for $\{F'_y\}_{y\in\Y}\in \overline{\G}_{(\alpha,\beta)}(\{\rho^x\}_{x\in\X})$, i.e., $(U^\dag B_y U)^\dag U^\dag B_y U=U^\dag B_y^\dag B_y U=U^\dag F_y U=F'_y$ for all $y\in\Y$. Therefore, $\{F'_y\}_{y\in\Y}\in\G'$. As a result, $\overline{\G}_{(\alpha,\beta)}(\{\rho^x\}_{x\in\X})\subseteq \G'$. 

    Combining Facts \#1 and \#2, we get $\overline{\G}_{(\alpha,\beta)}(\{\rho^x\}_{x\in\X})=\G'$. As a result, we may rewrite~\eqref{eqn:proof:2} to get
    \begin{align*}
        \overline{\mathcal{L}}_{ (\alpha,\delta)} &(X\rightarrow A)_{\U(\rho)}\\
        &\!= \!\!\!\!\sup_{\{F_y\}_y\in\overline{\G}_{(\alpha,\beta)}(\{\rho^x\}_{x\in\X})} \!\!\!\!\log_2\left(\sum_{y\in\mathbb{Y}} \max_{x\in\mathbb{X}} \trace(\rho^x F_y) \!\!\right)\!\!\\
        &\! =\overline{\mathcal{L}}_{(\alpha,\delta)}(X\rightarrow A)_{\rho}.
    \end{align*}
    This concludes the proof.
\end{proof}

For the weakly gentle quantum leakage, we may not be able to achieve unitary invariance. Instead, we prove a weak data-processing inequality that shows that the weakly gentle quantum leakage under application of the unitary operator can be bounded. This is done in the following proposition.

\begin{proposition}[Weak Data-Processing Inequality] \label{prop:dataprocessing}
    For any unitary channel $\U(\rho)=U\rho U^\dag$ with unitary operator $U\in\L(\H)$,
    \begin{align*}
        \mathcal{L}_{(\alpha,\delta)}(X\rightarrow A)_{\U(\rho)}\leq \mathcal{L}_{(\alpha+\beta(\U),\delta)}(X\rightarrow A)_{\rho},
    \end{align*}
    where $\beta(\U):=\sup_{\rho\in\{\rho^x\}_{x\in\X}} \|\U(\rho)-\rho\|_{\trace}.$
\end{proposition}

\begin{proof}
    Following a similar line of reasoning as in the proof of Proposition~\ref{prop:unitary_invariance}, we get
    \begin{align*}
        \mathcal{L}_{(\alpha,\delta)}&(X\rightarrow A)_{\U(\rho)}\\
        =&\!\!\!\!\sup_{\{F_y\}_y\in\G_{(\alpha,\delta)}(\{\U(\rho^x)\}_{x\in\X})} \log_2\left(\sum_{y\in\mathbb{Y}} \max_{x\in\mathbb{X}} \trace(\mathcal{E}(\rho^x) F_y) \!\!\right)\\
        =&\!\!\!\!\sup_{\{F_y\}_y\in\G''} \!\!\!\!\log_2\left(\sum_{y\in\mathbb{Y}} \max_{x\in\mathbb{X}} \trace(\rho^x F_y) \!\!\right)\!\!,
    \end{align*}
    where
    \begin{align*} 
        \G'':=\{\{F'_y\}_{y\in\Y}\,|\,&F'_y=U^\dag F_yU,\forall y\in\Y,\nonumber\\
        &\{F_y\}_{y\in\Y}\in\G_{(\alpha,\delta)}(\{\U(\rho^x)\}_{x\in\X})\}.
    \end{align*}
    Based on this, to prove the result, we have to show that $\G''\subseteq \G_{(\alpha+\beta(U),\delta)}(\{\rho^x\}_{x\in\X})$. To do so, pick any $\{F'_y\}_{y\in\Y}\in\G''$. There must exists $\{F_y\}_{y\in\Y}\in\G_{(\alpha,\delta)}(\{\U(\rho^x)\}_{x\in\X})$ such that $F'_y=U^\dag F_yU$ for all $y\in\Y$. Because $\{F_y\}_{y\in\Y}\in\G_{(\alpha,\delta)}(\{\U(\rho^x)\}_{x\in\X})$, there must exist an implementation $\B=\{B_y\}_{y\in\Y}$ such that $\mathbb{P}\left\{\|\rho_{\B|Y}-\rho\|_{\trace}\leq \alpha,\forall \rho\in S\right\}\geq 1-\delta$. Define $B'_y =B_y U$ for all $y\in\Y$. Evidently $\{B'_y\}_{y\in\Y}$ is an implementation of $\{F'_y\}_{y\in\Y}$ because ${B'_y}^\dag B'_y=(B_y U)^\dag (B_y U)=U^\dag F_y U=F'_y$ for all $y\in\Y$. For all $x\in\X$, we have
    \begin{align*}
         \hspace{.5in}&\hspace{-.5in} \left\|\frac{B'_y \rho^x {B'_y}^\dag}{\trace(B'_y\rho^x {B'_y}^\dag)} -\rho^x\right\|_{\trace}
         \\=& 
         \left\|\frac{(B'_y U^\dag) U \rho^x U^\dag (B'_y U^\dag)^\dag}{\trace((B'_y U^\dag) U \rho^x U^\dag (B'_y U^\dag)^\dag)} -\rho^x\right\|_{\trace}\\
         =& 
         \left\|\frac{(B'_y U^\dag) \U(\rho^x) (B'_y U^\dag)^\dag}{\trace(\U(\rho^x) UF'_y U^\dag)} -\rho^x\right\|_{\trace}\\
         \leq& 
        \left\|\frac{(B'_y U^\dag) \U(\rho^x) (B'_y U^\dag)^\dag}{\trace(\U(\rho^x)F_y)} -\U(\rho^x)\right\|_{\trace}\\
        &+\left\|\U(\rho^x) -\rho^x\right\|_{\trace}\\
        \leq &\alpha+\beta(\U),
    \end{align*}
    with probability larger than equal to $1-\delta$. Therefore, $\{F'_y\}_{y\in\Y}\in\G_{(\alpha+\beta(\U),\delta)}(\{\rho^x\}_{x\in\X})$, which means $\G''\subseteq \G_{(\alpha+\beta(\U),\delta)}(\{\rho^x\}_{x\in\X})$. 
    This concludes the result. 
\end{proof}

In the next proposition, we derive an upper bound for the gentle quantum leakage based on maximal quantum leakage. This demonstrates that, as expected, a discrete random variable $X\in\X$ has no more than $\log_2(|\X|)$ bits of information for an adversary to learn. In addition, the amount of leaked information is upper bounded by $2\log_2(\dim(\H))$. Hence, embedding or encoding classical information in a ``small'' quantum system, i.e., $\dim(\H)\ll |\X|$, can restrict the leakage (and perhaps simultaneously the utility of such encoding for information processing). 

\begin{proposition}[Upper Bound] \label{property:2}
     $\overline{\mathcal{L}}_{(\alpha,\delta)}(X\rightarrow A)_{\U(\rho)}\leq \mathcal{L}_{(\alpha,\delta)}(X\rightarrow A)_{\U(\rho)}\leq \min\{\log_2(|\X|),2\log_2(\dim(\H))\}$ if $\rho\in\S(\H)$.
\end{proposition}

\begin{proof}
    The proof follows from Corollary~\ref{cor:bound} and the upper bound for maximal quantum leakage in~\cite{Farokhi_PRA}.
\end{proof}

In the next proposition, we show that the gentle quantum leakage is reduced by application of a global depolarizing channel $\mathcal{D}_{p}:\S(\H)\rightarrow \S(\H)$ defined as
\begin{align} \label{eqn:dep_channel}
    \mathcal{D}_{p}(\rho):=\frac{p}{\dim(\H)}I+(1-p)\rho,
\end{align}
where $p\in[0,1]$ is a probability parameter capturing the magnitude of the depolarization ``noise''.  This implies that gentle quantum leakage accords with intuition (that quantum noise can reduce the information leakage) and follows similar results on privacy in quantum systems~\cite{hirche2022quantum,Farokhi_PRA}.

\begin{proposition} \label{prop:global}
For global depolarizing channel $\mathcal{D}_{p}$,
\begin{align*}
    \mathcal{L}_{(\alpha,\delta)}(X\rightarrow A)_{\mathcal{D}_{p}(\rho)}&=\log_2( p +(1-p)2^{\mathcal{L}_{(\alpha,\delta)}(X\rightarrow A)_{\rho}}),\\
    \overline{\mathcal{L}}_{(\alpha,\delta)}(X\rightarrow A)_{\mathcal{D}_{p}(\rho)}&=\log_2( p +(1-p)2^{\overline{\mathcal{L}}_{(\alpha,\delta)}(X\rightarrow A)_{\rho}}).
\end{align*}
Particularly,
    \begin{align*}
        \frac{\mathrm{d}}{\mathrm{d}p}\mathcal{L}_{(\alpha,\delta)}(X\rightarrow A)_{\mathcal{D}_{p}(\rho)}&<0 \mbox{ if } \mathcal{L}_{(\alpha,\delta)}(X\rightarrow A)_{\rho}>0,\\
        \frac{\mathrm{d}}{\mathrm{d}p}\overline{\mathcal{L}}_{(\alpha,\delta)}(X\rightarrow A)_{\mathcal{D}_{p}(\rho)}&<0 \mbox{ if } \overline{\mathcal{L}}_{(\alpha,\delta)}(X\rightarrow A)_{\rho}>0,
    \end{align*}
which shows that gentle quantum leakage is decreasing in probability parameter $p$.
\end{proposition}

\begin{proof} 
Following~\cite{Farokhi_PRA}, we know that    
\begin{align*}
     \sum_{y\in\mathbb{Y}} \max_{x\in\mathbb{X}} \trace(\mathcal{D}_{p}&(\rho^x)F_y)
     =  p +(1-p)\sum_{y\in\mathbb{Y}}\max_{x\in\mathbb{X}}\trace\left(\rho^x F_y\right).
\end{align*}
Therefore, 
\begin{align*}
        \mathcal{L}_{(\alpha,\delta)}(X&\rightarrow A)_{\mathcal{D}_{p}(\rho)}\\
        &=\log_2\bigg(\sup_{\{F_y\}_y}\sum_{y\in\mathbb{Y}} \max_{x\in\mathbb{X}} \trace(\mathcal{D}_{p}(\rho^x)F_y)\bigg)\\
        &=\log_2\bigg( p +(1-p)\sup_{\{F_y\}_y}\sum_{y\in\mathbb{Y}} \max_{x\in\mathbb{X}}\trace\left(\rho^x F_y\right)\!\!\bigg)\\
        &=\log_2( p +(1-p)2^{\mathcal{L}_{(\alpha,\delta)}(X\rightarrow A)_{\rho}}),
\end{align*}
where the supremum is taken over $\G_{(\alpha,\delta)}(\{\rho^x\}_{x\in\X})$. The proof for $\overline{\mathcal{L}}_{(\alpha,\delta)}(X\rightarrow A)_{\mathcal{D}_{p}(\rho)}$ follows from taking the supremum over $\overline{\G}_{(\alpha,\delta)}(\{\rho^x\}_{x\in\X})$. 
\end{proof}

\section{Approximate Cloning} \label{sec:cloning}
A quantum cloner is a quantum channel $T:\S(\H)\rightarrow \S(\H)\times \S(\H)$. This is often referred to as $1\rightarrow 2$ quantum cloning because 2 copies from 1 state are created. Ideally, we would like this mapping to be such that $T(\rho)=\rho\otimes\rho$ for any $\rho\in\S(\H)$. This is however forbidden by the rules of quantum system (particularly linearity of quantum channels)~\cite{wootters1982single}. Despite this negative result, imperfect or approximate cloning has been shown to be achievable~\cite{bruss1998optimal,buvzek1996quantum}. Let $T_i$ denote the quantum channel to the $i$-th copy of quantum state, i.e., $T_i(\rho)=\trace_{-i}(T(\rho))$, where $\trace_{-i}$ denotes the partial trace with respect all quantum systems except the $i$-th one. 

\begin{theorem}[Cloning Region~\cite{nechita2023asymmetric}] \label{tho:cloner} For all $p_1,p_2\in[0,1]$, there exists quantum cloner $T$ such that $T_i=\mathcal{D}_{p_i}$ if and only if
\begin{align*}
    \frac{d}{2}(d(2-p_1-p_2)+&\sqrt{d^2(p_1-p_2)^2-4(1-p_1)(1-p_2)})\\
    &-(2-p_1-p_2)\leq (d^2-1),
\end{align*}
where $d:=\dim(\H)$.
\end{theorem}

We can use this description of the achievable  approximate cloning to provide a lower bound for the weakly gentle quantum leakage. This is explored in the next proposition. 

\begin{proposition} \label{prop:lower_bound}
For all $\alpha,\delta\in[0,1]$, the weakly gentle quantum leakage is lower bounded by
\begin{align*}
    \mathcal{L}_{(\alpha,\delta)}(X\rightarrow A)_{\rho}\geq\mathcal{L}_{(\alpha,0)}(X\rightarrow A)_{\rho}\geq \underline{\mathcal{L}}_\alpha(\{\rho^x\}_{x\in\X}),
\end{align*}
where $\underline{\mathcal{L}}_\alpha(\{\rho^x\}_{x\in\X}):=\log_2(p_2^*+(1-p_2^*)2^{\mathcal{Q}(X\rightarrow A)_\rho})$ and $d=\dim(\H)$ with
    \begin{align*}
        (p_1^*,p_2^*)\in \min_{(p_1,p_2)}\quad  & p_2,\\
        \mathrm{s.t.} \quad & 0\leq p_1\leq 1,\\
        & 0\leq p_2\leq 1,\\
        & p_1 \leq \frac{\alpha}{\|{I}/{\dim(\H)}-\rho^x \|_{\trace}},\forall x\in\X,\\
        & 
        \begin{bmatrix}
            p_1 \\ p_2
        \end{bmatrix}^\top \!\!
        \begin{bmatrix}
         2d^2\!-\!1\!-\!\displaystyle\frac{d^4}{4}  & 
         1\!-\!\displaystyle\frac{d^4}{4} \\
         1\!-\!\displaystyle\frac{d^4}{4} & 
         2d^2\!-\!1\!-\!\displaystyle\frac{d^4}{4}   
        \end{bmatrix}\!\!
        \begin{bmatrix}
            p_1 \\ p_2
        \end{bmatrix}
        \\&\hspace{.2in}+
        \begin{bmatrix}
            -\!2\!-\!d^2 & -\!2\!-\!d^2
        \end{bmatrix}\!\!
        \begin{bmatrix}
            p_1 \\ p_2
        \end{bmatrix}+3\leq 0.
    \end{align*}
\end{proposition}

\begin{proof}
    The second copy can be used to extract secret information by Eve. By using the cloning technique provided in Theorem~\ref{tho:cloner}, 
    we can see that the second cloned state contains $\log_2(p_2+(1-p_2)2^{\mathcal{Q}(X\rightarrow A)_\rho})$ bits of information~\cite{Farokhi_PRA}. 
    Furthermore, the first copy can be passed to Bob for which the disturbance is given by $\|T_1(\rho)-\rho\|_{\trace}=p_1 \|{I}/{\dim(\H)}-\rho \|_{\trace}.$ Finally, the last constraint follows from the condition in Theorem~\ref{tho:cloner}.\end{proof}

\begin{remark}[Convexity]
    The optimization problem in Proposition~\ref{prop:lower_bound} is convex for $\dim(\H)= 2$. Therefore, for a qubit, i.e., when $\dim(\H)=2$, the lower bound can be easily computed. Unfortunately, the last constraint becomes non-convex if $\dim(\H)>2$.
\end{remark}

\begin{remark}[Incompatibility] A similar measure to $\underline{\mathcal{L}}_\alpha(\{\rho^x\}_{x\in\X})$ in Proposition~\ref{prop:lower_bound} has been proposed as a measure of mutual incompatibility of states within $\{\rho^x\}_{x\in\X}$, albeit for pure states and based on symmetric approximate cloning~\cite{mitra2021optimal}. Here, $\underline{\mathcal{L}}_\alpha(\{\rho^x\}_{x\in\X})$ generalizes this notion  of mutual incompatibility to asymmetric approximate cloning and potentially mixed states. It should be noted that if $\{\rho^x\}_{x\in\X}$ are mutually compatible, i.e., if $\rho^x$ and $\rho^{x'}$ commute for all $x,x'\in \X$, perfect cloning is possible and thus we can attain maximal information leakage with no state disturbance.
\end{remark}

\begin{corollary}
    For all $\alpha,\delta\in[0,1]$, $\mathcal{L}_{(\alpha,\delta)}(X\rightarrow A)_{\rho}=\mathcal{L}_{(\alpha,0)}(X\rightarrow A)_{\rho}=\mathcal{Q}(X\rightarrow A)_{\rho}$ if $[\rho^x,\rho^{x'}]=0$ for all $x,x'\in\X$. 
\end{corollary}

\begin{figure}
    \centering
    \begin{tikzpicture}
        \node[] at (0,0) {
        \psfragfig[width=.95\columnwidth]{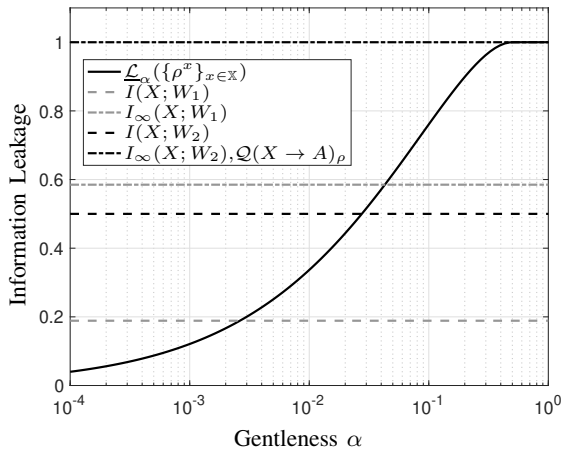}{
        }  
        };
        \node[] at (0,-3.1) {\small Gentleness $\alpha$};
        \node[rotate=90] at (-3.7,0) {\small Information Leakage};
    \end{tikzpicture}
    \vspace{-3mm}
    \caption{
    Information leakage and its lower bound versus weakly gentle measurement parameter $\alpha$. 
    }
    \vspace{-5mm}
    \label{fig:1}
\end{figure}

\section{Numerical Example} \label{sec:numerical}
Consider the BB84 scheme in~\cite{bennett2014quantum}. This is a scheme for quantum key distribution and its security relies on non-cloning theorem and post-measurement state collapse. Consider uniformly distributed random variable $X=(X_1,X_2)\in\{0,1\}^2$. The sender, i.e., Alice in Figure~\ref{fig:ABE}, encodes the random variable $X$ into a two-dimensional quantum system, i.e., a qubit, according to 
\begin{align*}
    \rho^x=
    \begin{cases}
        \ket{0}\bra{0}, & x=(0,0),\\
        \ket{1}\bra{1}, & x=(0,1),\\
        \ket{+}\bra{+}, & x=(1,0),\\
        \ket{-}\bra{-}, & x=(1,1),
    \end{cases}
\end{align*}
where $\{\ket{0},\ket{1}\}$ is the computational basis or the $Z$-basis, and $\{\ket{+},\ket{-}\}$ is the $X$-basis. Here, $X$ and $Z$ refer to Pauli matrices~\cite{nielsen2010quantum}. The qubit, for instance, can model the photon polarization when establishing quantum communication using single photons. Given that it is not possible to clone an arbitrary quantum state, the eavesdropper must measure the state directly or via an ancillary system interacting with it (e.g., based on approximate cloning). For instance, the eavesdropper can toss a coin and, based on the outcome, measure the state either in the $Z$-basis or the $X$-basis. This measurement is denoted by $W_1$ in what follows. Another approach is to always measure in the $X$-basis. This measurement is denoted by $W_2$ in what follows. Both these measurement strategies result in a non-trivial change of the state (as they are projection based), which the sender or the receiver can use to identify the presence of a malicious party and thus abandon communication (i.e., key sharing protocol) with high probability. This was used to establish the safety of the BB84 algorithm. Instead of the proposed measurement strategies, the eavesdropper can use gentle measurements proposed in this paper to extract some information, albeit not the maximum possible amount, while reducing the detection probability. Note that as $\alpha$ and $\delta$ (i.e., the gentleness) vary the amount of extracted information changes but also the probability of detection by the sender or receiver. 

Figure~\ref{fig:1} illustrates information leakage versus weakly gentle measurement parameter $\alpha$. The solid curve illustrates the lower bound $\underline{\mathcal{L}}_\alpha(\{\rho^x\}_{x\in\X})$ in Proposition~\ref{prop:lower_bound}. Other measures of information leakage that do not rely on gentle measurement are also demonstrated for comparison. 
The maximum amount of information that can be extracted, captured by the maximal quantum leakage $\mathcal{Q}(X\rightarrow A)_\rho$ and attained by measurement $W_2$, is 1~bit. As we increase $\alpha\rightarrow 1$, the lower bound gets closer to the maximal quantum leakage. As we reduce $\alpha\rightarrow 0$, the lower bound  reduces. It is interesting to that for a relatively small $\alpha=0.1$ (which reduces the chance of Eve getting caught significantly), the eavesdropper can still steal a significant amount of information $\mathcal{L}_{(0.1,\delta)}(X\rightarrow A)_{\rho}\geq \underline{\mathcal{L}}_{0.1}(\{\rho^x\}_{x\in\X})=0.7608$ for any $\delta\in[0,1]$.

\section{Conclusions and Future Work} \label{sec:conc}
We developed and studied a new notion for quantum information leakage against arbitrary eavesdroppers that aim to avoid detection. Gentle measurements were used to encode the desire of the adversary to evade detection. Future work can focus on proving data-processing inequality for general quantum channels and (sub-)additivity when accessing multiple quantum systems. Also, a numerical algorithm for computing the gentle quantum leakage can be developed.

\bibliography{ref}

\begin{thebibliography}{10}

\bibitem{bennett2014quantum}
C.~H. Bennett and G.~Brassard, ``Quantum cryptography: Public key distribution
  and coin tossing,'' {\em Theoretical computer science}, vol.~560, pp.~7--11,
  2014.

\bibitem{nielsen2010quantum}
M.~Nielsen and I.~Chuang, {\em Quantum Computation and Quantum Information}.
\newblock Cambridge University Press, 2000.

\bibitem{ekert1994eavesdropping}
A.~K. Ekert, B.~Huttner, G.~M. Palma, and A.~Peres, ``Eavesdropping on
  quantum-cryptographical systems,'' {\em Physical Review A}, vol.~50, no.~2,
  p.~1047, 1994.

\bibitem{fuchs1996quantum}
C.~A. Fuchs and A.~Peres, ``Quantum-state disturbance versus information gain:
  Uncertainty relations for quantum information,'' {\em Physical Review A},
  vol.~53, no.~4, p.~2038, 1996.

\bibitem{fuchs1998information}
C.~A. Fuchs, ``Information gain vs. state disturbance in quantum theory,'' {\em
  Fortschritte der Physik: Progress of Physics}, vol.~46, no.~4-5,
  pp.~535--565, 1998.

\bibitem{fuchs2001information}
C.~A. Fuchs and K.~Jacobs, ``Information-tradeoff relations for finite-strength
  quantum measurements,'' {\em Physical Review A}, vol.~63, no.~6, p.~062305,
  2001.

\bibitem{banaszek2006information}
K.~Banaszek, ``Information gain versus state disturbance for a single qubit,''
  {\em Open Systems \& Information Dynamics}, vol.~13, no.~1, pp.~1--16, 2006.

\bibitem{issa2019operational}
I.~Issa, A.~B. Wagner, and S.~Kamath, ``An operational approach to information
  leakage,'' {\em IEEE Transactions on Information Theory}, vol.~66, no.~3,
  pp.~1625--1657, 2019.

\bibitem{Farokhi_PRA}
F.~Farokhi, ``Maximal information leakage from quantum encoding of classical
  data,'' {\em Phys. Rev. A}, vol.~109, p.~022608, Feb 2024.

\bibitem{liao2019tunable}
J.~Liao, O.~Kosut, L.~Sankar, and F.~du~Pin~Calmon, ``Tunable measures for
  information leakage and applications to privacy-utility tradeoffs,'' {\em
  IEEE Transactions on Information Theory}, vol.~65, no.~12, pp.~8043--8066,
  2019.

\bibitem{farokhi2021measuring}
F.~Farokhi and N.~Ding, ``Measuring information leakage in non-stochastic
  brute-force guessing,'' in {\em 2020 IEEE Information Theory Workshop (ITW)},
  pp.~1--5, IEEE, 2021.

\bibitem{aaronson2019gentle}
S.~Aaronson and G.~N. Rothblum, ``Gentle measurement of quantum states and
  differential privacy,'' in {\em Proceedings of the 51st Annual ACM SIGACT
  Symposium on Theory of Computing}, pp.~322--333, 2019.

\bibitem{7308959}
S.~Verdu, ``$\alpha$-mutual information,'' in {\em 2015 Information Theory and
  Applications Workshop (ITA)}, pp.~1--6, 2015.

\bibitem{muller2013quantum}
M.~M{\"u}ller-Lennert, F.~Dupuis, O.~Szehr, S.~Fehr, and M.~Tomamichel, ``On
  quantum {R}{\'e}nyi entropies: A new generalization and some properties,''
  {\em Journal of Mathematical Physics}, vol.~54, no.~12, 2013.

\bibitem{10106314906367}
K.~M.~R. Audenaert and N.~Datta, ``{$\alpha$-$z$-{R}\'{e}nyi relative
  entropies},'' {\em Journal of Mathematical Physics}, vol.~56, no.~2,
  p.~022202, 2015.

\bibitem{PhysRevE104035309}
J.~Vovrosh, K.~E. Khosla, S.~Greenaway, C.~Self, M.~S. Kim, and J.~Knolle,
  ``Simple mitigation of global depolarizing errors in quantum simulations,''
  {\em Phys. Rev. E}, vol.~104, p.~035309, Sep 2021.

\bibitem{farokhi2024barycentric}
F.~Farokhi, ``Barycentric and pairwise r\'{e}nyi quantum leakage,'' {\em arXiv
  preprint arXiv:2402.06156}, 2024.

\bibitem{nuradha2023quantum}
T.~Nuradha, Z.~Goldfeld, and M.~M. Wilde, ``Quantum pufferfish privacy: A
  flexible privacy framework for quantum systems,'' {\em arXiv preprint
  arXiv:2306.13054}, 2023.

\bibitem{mitra2021optimal}
A.~Mitra and P.~Mandayam, ``On optimal cloning and incompatibility,'' {\em
  Journal of Physics A: Mathematical and Theoretical}, vol.~54, no.~40,
  p.~405303, 2021.

\bibitem{wilde2013quantum}
M.~Wilde, {\em Quantum Information Theory}.
\newblock Quantum Information Theory, Cambridge University Press, 2013.

\bibitem{wiseman2010quantum}
H.~Wiseman and G.~Milburn, {\em Quantum Measurement and Control}.
\newblock Cambridge University Press, 2010.

\bibitem{Benedetti_2019}
M.~Benedetti, E.~Lloyd, S.~Sack, and M.~Fiorentini, ``Parameterized quantum
  circuits as machine learning models,'' {\em Quantum Science and Technology},
  vol.~4, no.~4, p.~043001, 2019.

\bibitem{hirche2022quantum}
C.~Hirche, C.~Rouz\'{e}, and D.~S. Franca, ``Quantum differential privacy: An
  information theory perspective,'' {\em IEEE Transactions on Information
  Theory}, vol.~69, no.~9, pp.~5771--5787, 2023.

\bibitem{wootters1982single}
W.~K. Wootters and W.~H. Zurek, ``A single quantum cannot be cloned,'' {\em
  Nature}, vol.~299, no.~5886, pp.~802--803, 1982.

\bibitem{bruss1998optimal}
D.~Bru{\ss}, D.~P. DiVincenzo, A.~Ekert, C.~A. Fuchs, C.~Macchiavello, and
  J.~A. Smolin, ``Optimal universal and state-dependent quantum cloning,'' {\em
  Physical Review A}, vol.~57, no.~4, p.~2368, 1998.

\bibitem{buvzek1996quantum}
V.~Bu{\v{z}}ek and M.~Hillery, ``Quantum copying: Beyond the no-cloning
  theorem,'' {\em Physical Review A}, vol.~54, no.~3, p.~1844, 1996.

\bibitem{nechita2023asymmetric}
I.~Nechita, C.~Pellegrini, and D.~Rochette, ``The asymmetric quantum cloning
  region,'' {\em Letters in Mathematical Physics}, vol.~113, no.~3, p.~74,
  2023.

\end{thebibliography}
\bibliographystyle{ieeetr}

\appendices

\section{Useful Lemma}
\label{app:useful_lemma}
\begin{lemma} \label{lemma:gentle_construct}
    Consider 
    \begin{align*}
        B_+&=\sqrt{\frac{1-2\epsilon^2}{2}}I+\epsilon M,\\
        B_-&=\sqrt{\frac{1-2\epsilon^2}{2}}I-\epsilon M,\\
        B_0&=\sqrt{2}\epsilon(I- M^2)^{1/2},
    \end{align*}
    for semi-positive definite operator $M$. Then $\{F_+,F_-,F_0\}$ with $F_y=B_yB_y^\dag$ with $y\in\{+,-,0\}$ forms a POVM. For any $\alpha>0$, there exists $\epsilon'>0$ such that, for all $0\leq \epsilon\leq \epsilon'$, $\{F_+,F_-,F_0\}$ is $(\alpha,\delta)$-weakly gentle on any $S\subseteq\S(\H)$.
\end{lemma}

\begin{proof}
    Note that $F_0=B_0B_0^\dag\geq 0$ because $M\leq I$. Also,
    \begin{align*}
    F_-
    =&B_-B_-^\dag
    =\frac{1-2\epsilon^2}{2}I-2\epsilon\sqrt{\frac{1-2\epsilon^2}{2}}M+\epsilon^2M.
    \end{align*}
    Therefore, $F_-\geq 0$ if $\epsilon\leq 1/10$.   
    Furthermore, 
    \begin{align*}
        F_++F_-+F_0
        =&B_+B_+^\top+B_-B_-^\top+B_0B_0^\top\\
        =&\frac{1-2\epsilon^2}{2}I+\sqrt{2(1-2\epsilon^2)}\epsilon M+\epsilon^2 M^2\\
        &+\frac{1-2\epsilon^2}{2}I-\sqrt{2(1-2\epsilon^2)}\epsilon M+\epsilon^2 M^2\\
        &+2\epsilon^2 (I-M^2)\\
        =&I.
    \end{align*}
    Therefore, $\{F_+,F_-,F_0\}$ is a POVM. We have
    \begin{align*}
       B_+\rho B_+^\dag
       =&\left(\sqrt{\frac{1-2\epsilon^2}{2}}I+\epsilon M \right)\rho \left(\sqrt{\frac{1-2\epsilon^2}{2}}I+\epsilon M \right)\\
       =&\frac{1-2\epsilon^2}{2}\rho +\epsilon\sqrt{\frac{1-2\epsilon^2}{2}}(M\rho+\rho M)+\epsilon^2 M\rho M
    \end{align*}
    and
    \begin{align*}
        \trace(B_+\rho B_+^\dag)
        =& \frac{1-2\epsilon^2}{2}+\epsilon\sqrt{2(1-2\epsilon^2)}\trace(\rho M)\\
        &+\epsilon^2 \trace(M\rho M)\\
        =&\frac{1-2\epsilon^2}{2}+\epsilon\sqrt{2(1-2\epsilon^2)}\trace(\rho M)+\mathcal{O}(\epsilon^2).
    \end{align*}
    As a result,
    \begin{align*}
        \frac{B_+\rho B_+^\dag}{\trace(B_+\rho B_+^\dag)}
        =&\frac{\frac{1-2\epsilon^2}{2}\rho +\epsilon\sqrt{\frac{1-2\epsilon^2}{2}}(M\rho+\rho M)+\mathcal{O}(\epsilon^2)}{\frac{1-2\epsilon^2}{2}+\epsilon\sqrt{2(1-2\epsilon^2)}\trace(\rho M)+\mathcal{O}(\epsilon^2)}\\
        =&\rho+\sqrt{\frac{2}{1-2\epsilon^2}}\epsilon(M\rho+\rho M+2\trace(\rho M)\rho)\\&
        +\mathcal{O}(\epsilon^2).
    \end{align*}
    This shows that
    \begin{align*}
        &\left\|\frac{B_+\rho B_+^\dag}{\trace(B_+\rho B_+^\dag)}-\rho \right\|_{\trace}
        \\&\quad=\epsilon\sqrt{\frac{2}{1-2\epsilon^2}}
        \|M\rho+\rho M+2\trace(\rho M)\rho\|_{\trace}+\mathcal{O}(\epsilon^2).
    \end{align*}
    Therefore, there exists small non-zero $\epsilon_+$ such that, for all $0\leq \epsilon\leq \epsilon_+$, 
    \begin{align*}
        \left\|\frac{B_+\rho B_+^\dag}{\trace(B_+\rho B_+^\dag)}-\rho \right\|_{\trace}\leq \alpha.
    \end{align*}
    Following the same line of reasoning, we can demonstrate that there exists small non-zero $\epsilon_-$ such that, for all $0\leq \epsilon\leq \epsilon_-$, 
    \begin{align*}
        \left\|\frac{B_-\rho B_-^\dag}{\trace(B_-\rho B_-^\dag)}-\rho \right\|_{\trace}\leq \alpha.
    \end{align*}
    Finally, note that, if $\epsilon\leq \sqrt{\delta/(2(1-\trace(M^2\rho))}$, we get
    \begin{align*}
        \trace(B_+\rho B_+^\dag)
        +\trace(B_-\rho B_-^\dag)
        =&1-2\epsilon^2 (1-\trace(M^2\rho))\\
        \geq& 1-\delta.
    \end{align*}
    This proves the lemma with $$\epsilon'=\min(\epsilon_+,\epsilon_-,\sqrt{\delta/(2(1-\trace(M^2\rho))},1/10).
    $$
\end{proof}

\end{document}